\documentclass[a4paper, 12pt]{article}

\RequirePackage{amsthm,amsmath,amsfonts,amssymb,booktabs}
\RequirePackage[authoryear]{natbib}
\RequirePackage{hyperref}
\RequirePackage{graphicx}
\usepackage{fullpage}
\usepackage{epstopdf}
\usepackage[caption=false]{subfig}
\usepackage[doublespacing]{setspace}
\usepackage{url}

\usepackage[authoryear]{natbib}

\theoremstyle{plain}
\newtheorem{theorem}{Theorem}[section]

\newtheorem{proposition}[theorem]{Proposition}

\theoremstyle{definition}

\theoremstyle{remark}

\newcommand{\unif}{{\sf Unif}}
\newcommand{\nm}{{\sf N}}

\newcommand{\gam}{{\sf Gamma}}

\newcommand{\sign}{\mathrm{sign}}

\newcommand{\RR}{\mathbb{R}}

\newcommand{\GG}{\mathbb{G}}

\newcommand{\eps}{\varepsilon}

\theoremstyle{plain}

\newtheorem{lem}{Lemma}

\newtheorem{cond*}{Condition}
\newtheorem{assump}{Assumption}[section]

\title{Robust posterior inference for Youden's index cutoff}
\author{Nicholas Syring\footnote{Department of Statistics, Iowa State University; {\tt nsyring@iastate.edu}. }}
\date{\today}

\begin{document}

\maketitle 

\begin{abstract}
Youden's index cutoff is a classifier mapping a patient's diagnostic test outcome and available covariate information to a diagnostic category.  Typically the cutoff is estimated indirectly by first modeling the conditional distributions of test outcomes given diagnosis and then choosing the optimal cutoff for the estimated distributions.  Here we present a Gibbs posterior distribution for direct inference on the cutoff.  Our approach makes incorporating prior information about the cutoff much easier compared to existing methods, and does so without specifying probability models for the data, which may be misspecified.  The proposed Gibbs posterior distribution is robust with respect to data distributions, is supported by large-sample theory, and performs well in simulations compared to alternative Bayesian and bootstrap-based methods.  In addition, two real data sets are examined which illustrate the flexibility of the Gibbs posterior approach and its ability to utilize direct prior information about the cutoff.

\smallskip

\end{abstract}

\section{Introduction} 

Patient diagnosis is challenging and subject to many sources of random variation, e.g., physician discretion, diagnostic test variability, and patient characteristics to name a few. Youden's index and its associated cutoff value are valuable tools that help doctors determine the best diagnostic testing procedure and evaluate its effectiveness. These tools are themselves subject to sampling variability so methods for statistical inference on Youden's index and cutoff could be helpful input to physicians making diagnostic decisions. 

For example, consider the relationship between performance on the Trail Making Test (TMT) and cognitive impairment (see Section~5).  The TMT is a two-part diagnostic test requiring test takers to draw lines connecting objects in the appropriate order, correctly associating alphanumeric labels.  Both the time it takes to complete the TMT and the number of errors made are helpful in determining a patient's level of cognitive impairment; see \citet{rasmussen.etal.1998}, \citet{amieva.etal.1998}, and \citet{ashendorf.etal.2008}.  Youden's index and cutoff can precisely quantify the relationship between TMT performance and cognitive impairment and may improve the accuracy of diagnoses. 

The cutoff value associated with Youden's index is a value of the diagnostic test that acts as a classifier, categorizing patients by diagnostic outcome into diagnostic categories.  Classifiers can be consistently estimated by minimizing empirical misclassification error, but such procedures have no ability to incorporate prior information, which may be well developed in a diagnostic setting.  When prior information is available a Bayesian approach is preferred, except classifiers are not easily modeled as parameters of probability distributions.  Convenient models---like logistic or probit regression---may not contain the true data-generating model, and Bayesian posterior distributions built from such misspecified models may produce biased inferences.

The literature on statistical inference for Youden's index and cutoff is plentiful, but so far has not adequately addressed the above issues.  Both parametric and nonparametric methods for point estimation of Youden's index were studied in \citet{fluss.etal.2005} using simulations.  \citet{xu.etal.2014} proposed a covariate-adjusted nonparametric M-estimation technique for Youden's index and cutoff.  Interval estimates have been addressed by many authors.  Both \citet{lai.etal.2012} and \citet{shan.2015} considered improved parametric interval estimates when the data is randomly sampled from a normal distribution.  \citet{yin.etal.2016} investigated improvements to interval estimation by utilizing a more advanced sampling methodology.  \citet{molanes.2011} introduced a nonparametric empirical likelihood method with good performance for confidence intervals, but their method is computationally intensive, requiring the use of the bootstrap, optimization, and kernel-smoothing steps.  \citet{wang.etal.2017} improved upon the empirical likelihood method and showed their resulting confidence intervals have better coverage probability and shorter average length than percentile bootstrap intervals.  All of the above references considered only two diagnostic categories whereas \citet{nakas.etal.2010} and \citet{nakas.etal.2013} defined extensions of Youden's index and cutoff to general ordered multi-group classification problems.  \citet{carvalho.2018} used a nonparametric Bayesian method for inference on Youden's index and cutoff for three diagnostic categories.  Despite its name, this method actually does model the data distributions, and their use of a nonparametric model means that, despite being a Bayesian method, their prior distribution is not able to be used to incorporate real prior information about Youden's index.  \citet{carvalho.2017} use a similar nonparametric Bayesian method for the covariate-adjusted two-category problem.   

We develop a Gibbs posterior distribution for inference on the cutoff in order to address the shortcomings of existing methods highlighted above.  The Gibbs posterior methodology combines the robustness of M-estimation with the Bayesian approach's seamless integration of prior information.  Our method accommodates covariates and any number of diagnostic categories.  We show our Gibbs posterior distribution is consistent and concentrates on the true cutoff at a fast rate without needing to correctly specify the data-generating mechanism.  And, we demonstrate the Gibbs posterior distribution's good performance relative to M-estimation and Bayesian inference in both simulation examples and two real data examples.  

A brief outline of this paper is as follows: Section~2 defines Youden's index and cutoff for $k$ diagnostic categories and discusses how to accommodate covariate information; Section~3 discusses nonparametric point estimation of the cutoff and introduces the Gibbs posterior distribution for the cutoff in order to incorporate prior information about the cutoff; Section~4 briefly highlights the theoretical support for inferences based on the Gibbs posterior distribution; Section~5 presents real data examples; Section~6 provides simulation examples; and, Section~7 makes concluding remarks.  Proofs of the results in Section~4 are presented in Appendices A and B.  Data and R codes for the examples presented in Section~5 can be accessed at \url{https://github.com/nasyring/GPC_YI}.

\section{Youden's index and cutoff}
\subsection{Multi-class setting}
\label{SS:YI}
Suppose diseased $(Y=1)$ and healthy $(Y=-1)$ individuals predictably respond to a diagnostic test with the diseased individuals tending to have larger continuous outcomes $X$. Mathematically, this means the conditional distributions of the test outcome given health are stochastically ordered; that is, $F_{-1}(x)>F_1(x)$ $\forall x$, where $F_{j}(x) = P(X\leq x|Y=j)$, $j=-1, 1$.  The purpose of the diagnostic test is to classify patients, or in other words, to test the hypothesis $H_0:Y=1$ versus $H_a:Y = -1$ by separating the diagnostic categories according to $Y = -1$ for $X <\theta$ and $Y = 1$ for $X>\theta$ based on an optimal cutoff $\theta$.  Denote the joint distribution of $(X,Y)$ by $P$, and in the sequel denote the product measure of a random sample of pairs $(X_i, Y_i)$ for $i=1, \ldots, n$ by $P^n$. The \textit{sensitivity} of the test is $1-F_{1}(\theta):=P(X > \theta,|\,Y=1)$ and the \textit{specificity} of the test is $F_{-1}(\theta) = P(X\leq\theta\,|\,Y=-1)$.  When Type I and Type II errors are equally important (or equally costly) then the best test is the one maximizing the sum of sensitivity and specificity.  That maximum is called Youden's index and the corresponding cutoff for the test is termed Youden's cutoff, or simply, the cutoff.  In this paper we consider only estimation and inference for the cutoff, the parameter to be used for diagnosis.  

In the multi-class setting there are $k$ diagnostic categories, $Y \in \{1,2,...,k\}$, and (possibly unequal) weights on Type I and II errors.  For an ordered $(k-1)-$dimensional vector ${\theta}=(\theta_1, \ldots, \theta_{k-1})^\top$ with $\theta_1<\cdots <\theta_{k-1}$ Youden's index and cutoff are defined as   
\begin{align*}
YI &= \max_{{\theta}} \sum_{j=1}^{k-1}w_jF_j(\theta_j)- (1-w_j)F_{j+1}(\theta_j);
\end{align*}
and, 
\begin{align}
\label{eq:R_fcn}
{\theta^\star}&=\arg\min_{{\theta}}\sum_{j=1}^{k-1}(1-w_j)F_{j+1}(\theta_j)-w_jF_j(\theta_j),
\end{align}
where $w_j \in (0,1)$; see also \citet{nakas.etal.2013}.  In the remainder we will assume the weights are equal and so may be omitted, but the statistical inference methods we discuss are easily extended to the situation with unequal weights.

The conditional probabilities $F_{j}(\theta):=P(X \leq \theta,\,Y=j)/P(Y=j)$ have two common interpretations, which have important consequences when it comes to estimation of the cutoff; see also \citet{xu.etal.2014} and Section~4. Let $1(\cdot)$ denote the indicator function and consider observed diagnoses $y_1, \ldots, y_n$.  Define $p_j:=P(Y=j)$.  Then, $p_j$ is either known or is estimated by the sample proportion $\hat p_j = \frac{1}{n}\sum_{i=1}^n 1(y_i = j)$. In case-control studies $p_j$ is known because the proportions of participants falling into each diagnostic category are set in advance, whereas in cohort studies the proportion of patients eventually diagnosed to category $j$ is not known when the experiment begins.    
\subsection{Covariate-adjusted cutoff}
\label{S:covariate}

Often times diagnostic results are strongly associated with patient characteristics in addition to diagnostic group. For example, the level of blood pressure reading indicating hypertension depends on age.  In order to account for covariate influence on diagnostic results we define a covariate-adjusted cutoff for $k=2$ diagnostic groups  analogously to our definition in Section~2.1.  The cutoff function $\theta^\star(z)$ of a continuous, vector covariate $z\in \mathbb R^p$ satisfies
\[\theta^\star( z) = \arg_\theta\min \{F_{1}(\theta( z))-F_{-1}(\theta( z))\}\]
for every $z$.  

There are many models of covariate information that could be used, including linear models $\theta^\star( z) =  z^\top \beta$ and additive models $\theta^\star( z) = \sum_{j=1}^p f_j(z_j)$ for smooth functions $f_j$, $j=1,\ldots,p$.  One could also consider a tree-based model, such as the two-covariate model $\theta^\star(z_1, z_2) = \theta_{ij}$ if $(z_1, z_2)\in (z_{1,i}, z_{1,i+1}]\times(z_{2,j}, z_{2,j+1}]$ for $i=1,\ldots,I$ and $j=1, \ldots, J$, and where $z_{1,1},\ldots,z_{1,I}$ and $z_{2,1},\ldots,z_{2,J}$ are splits of the covariates $z_1$ and $z_2$ partitioning the covariate space.  Our approach can accommodate any of these models of covariate information.  Here, we focus on a smooth function model $\theta^\star(z)$ for a single covariate $z\in [0,1]$, which has been used in practical applications; and see the examples in Sections 5.2 and 6.2.  For a specific notion of smoothness, consider the $\alpha-$H\"older smooth functions $\theta:[0,1]\mapsto \mathbb{R}$ satisfying for all $z,z' \in [0,1]$
\[ |\theta^{\star ([\alpha])}(z) - \theta^{\star ([\alpha])}(z')| \leq L |z - z'|^{\alpha - [\alpha]}, \]
where the superscript ``$(k)$'' means $k^\text{th}$ derivative and $[\alpha]$ is the integer part of $\alpha$.

In order to define a prior and Gibbs posterior distribution we need to choose a parametrization of $\theta(z)$.  It is convenient to represent $\theta(z)$ as a linear combination of basis functions, and in the examples that follow we use cubic b-splines defined as follows: 
\[\theta(z):=\sum_{j=1}^d \beta_j B_{j,d}(z),\]
where $d$ denotes the number of basis functions used, $B_{j,d}(z)$ denotes each of those $d$ basis functions for $j=1, \ldots, d$, and $\beta_j\in\mathbb{R}$ denotes the coefficient of the $j^{th}$ basis function in the linear combination.  For $d$ basis functions, a cubic b-spline is defined on a set of $d+4$ knots, denoted $t_{-3}\leq t_{-2}\leq t_{-1}\leq t_0=0< ...<t_{d-3}=1\leq t_{d-2}\leq t_{d-1}\leq t_d$, defining a grid in $\mathbb{R}$. Eight knots lie outside $(0,1)$ while $d-4$ knots line inside $(0,1)$.  For $z\in (0,1)$ the b-spline basis functions are defined recursively by
\begin{align*}
&B_{j,1}(z) = 1(z\in[t_j, t_{j+1}])\quad\text{and}\quad\\
&B_{j,\ell}(z) = \frac{z-t_j}{t_{j+\ell-1}-t_j}B_{j,\ell-1}(z)
+\frac{t_{j+\ell}-z}{t_{j+\ell}-t_{j+1}}B_{j+1,\ell-1}(z),
\end{align*}
for $\ell=2,\ldots,d$.

Assuming $\theta^\star(z)$ is an $\alpha-$H\"older smooth function it can be well--estimated by a linear combination of cubic b-splines.  Let ${\beta} := (\beta_1,...,\beta_d)^\top$ and ${B_d(z)} := (B_{1,d}(z),....B_{d,d}(z))^\top$.  Then, there exists a constant $H>0$ such that for every $d>0$ there exists a linear combination ${\theta(z)} = {\beta^\top B_d(z)}$ such that
\begin{equation}
    \label{eq:approx.smooth}
    \|\beta\|_\infty < H\text{ and }\|\theta - \theta^\star\|_\infty \leq Cd^{-\alpha},
\end{equation}
for some constant $C>0$.  \eqref{eq:approx.smooth} says that $\theta^\star$ can be well-approximated by a linear combination of cubic b-spline functions bounded by $H$.  Boundedness of ${\beta}$ is helpful, for example, because it means the prior distribution over coefficient vectors need not be too spread out.

\section{Robust posterior inference on the cutoff}
\label{S:estimation}
\subsection{Objective function}

Following our point of view laid out in the introduction, we now turn to M-estimation methods that avoid specifying probability models for the conditional distribution functions $F_j(\theta)$ of diagnostic outcomes.  Misspecification of these distributions can cause estimates to be biased, but the strategy presented in this section avoids bias by avoiding probability-model specification altogether. 

In the multi-class setting we estimate the cutoff ${\theta^\star}$ by the minimizer ${\hat\theta_n} := \arg\min R_n({\theta})$ of an objective function $R_n({\theta};Y,X)$ (abbreviated to $R_n({\theta})$).
The simplest choice of $R_n({\theta})$ is the empirical version of \eqref{eq:R_fcn} (omitting the equal weights),
\begin{equation}\label{eq:empirical_Rn}\begin{aligned}
R_n({\theta}) &= \frac{1}{n}\sum_{i=1}^n \sum_{j=1}^{k-1}\ell(\theta_j;y_i,x_i) \\
& = \frac{1}{n}\sum_{i=1}^n \sum_{j=1}^{k-1}\Biggl\{\frac{1(x_i\leq \theta_j, y_i=j+1)}{ p_{j+1}}-\frac{1(x_i\leq \theta_j, y_i=j)}{ p_j}\Biggr\}
\end{aligned}
\end{equation}
where every diagnostic class is assumed to appear in the data.  

In the covariate-adjusted setting for $k=2$ diagnostic categories we estimate $\theta^\star$ by the minimizer of the corresponding empirical objective function for a given number of basis functions $d$.  The minimizer, denoted $\hat\theta_n(z) :=$ ${\hat{{\beta}}_n}^\top{B_{d}(z)}$, is a linear combination of b-spline basis functions and satisfies
\begin{equation}
\label{eq:empirical_Rn_func}
\begin{aligned}
{\hat{{\beta}}_n} &:= \arg\min_{{\beta}}R_n({\beta}) \\
&:= \arg\min_{{\beta}}\frac1n\sum_{i=1}^n\ell({\beta}^\top {B_d(z_i)};y_i,x_i,z_i) \\
&=\arg\min_{{\beta}}  \frac{1}{n}\sum_{i=1}^n \Biggl\{\frac{1(x_i\leq {\beta}^\top {B_d(z_i)}, y_i=1)}{ p_{1}}-\frac{1(x_i\leq {\beta}^\top {B_d(z_i)}, y_i=-1)}{ p_{-1}}\Biggr\}.
\end{aligned}
\end{equation}
\subsection{Gibbs posterior distribution for the cutoff}
\label{S:gibbs}

Our primary goal is to provide some reliable summary of uncertainty about the cutoff ${\theta^\star}$.  One strategy towards achieving this is to estimate the cutoff by minimizing the objective function and, further, to bootstrap that process to compute a confidence region for ${\theta^\star}$.  This strategy has the advantage of utilizing the objective function in \eqref{eq:empirical_Rn} that defines the cutoff, instead of relying on a (possibly misspecified) probability model, but it cannot handle prior information.  If we have an informative prior distribution for the cutoff then we should combine that information with the data, and the Bayesian method is the usual strategy for accomplishing that combination.  To define a Bayesian posterior distribution for the cutoff we need a likelihood function, which is defined by selecting probability models for the observed diagnostic measurements in each category.  There are a couple of difficulties with this approach.  First, our likelihood will certainly be a function of some model parameters, and not directly of the cutoff ${\theta^\star}$, so any prior information we have about the cutoff would have to be translated into prior information about model parameters, and it is not obvious how that would be done, if it even can be.  Second, since these diagnostic data distributions are effectively nuisance parameters it is at least wasteful that we should have to model them, and, at worst, if we model them incorrectly we may incur bias in our inferences about ${\theta^\star}$. 

An alternative strategy is to define a Gibbs posterior distribution for the cutoff follows: the Gibbs posterior probability of a measureable set $A\subset \Theta$ is
\begin{equation}
\label{eq:gibbs}
\Pi_{n,\omega_n}(A) = \frac{\int_A \exp[-\omega_n nR_{n}({\theta})]\Pi(d{\theta})}{\int_\Theta \exp[-\omega_n nR_{n}({\theta})]\Pi(d{\theta})}
\end{equation}
for prior distribution $\Pi$, objective function $R_n({\theta})$, and a sequence $\omega_n>0$---often called the \emph{learning rate}---that provides a weighting of the objective function relative to the prior distribution.  The Gibbs posterior distribution uses a pseudo-likelihood function equal to an exponential transformation of the objective function in \eqref{eq:empirical_Rn}, so there is no need to model the data distributions.  Moreover, the only parameter appearing in the Gibbs posterior distribution is the cutoff ${\theta}$, so no marginalization or prior specification for nuisance model parameters is necessary.  There has been substantial work on pseudo-posterior distributions, of which Gibbs posteriors are a special case, but for this application the Gibbs posterior distribution, in particular, is a principled choice for inference.  Suppose $\Pi_n$ denotes a Bayesian posterior for ${\theta}$.  We would expect posterior draws ${\theta}\sim \Pi_n$ would, on average, nearly maximize the likelihood used to construct the posterior.  On the other hand, since ${\theta^\star}$ is defined as the minimizer of the expectation of the objective function given in \eqref{eq:empirical_Rn}, any reasonable posterior for ${\theta}$ should, on average, minimize $R_n({\theta};{y},{x})$, given data and posterior draws ${\theta}\sim \Pi_{n,\omega_n}$.     \citet{bissiri.2016} and \citet{zhang.2006} show the Gibbs posterior in \eqref{eq:gibbs} does exactly that---it minimizes the posterior expectation of the objective function over the class of measures defined on $\Theta$.  Further, \citet{bissiri.2016} show the Gibbs posterior is the unique solution to this minimization problem if we consider only coherent posterior distributions.  In that sense, the Gibbs posterior defined in \eqref{eq:gibbs} is the unique form of pseudo-posterior distribution based on the objective function in \eqref{eq:empirical_Rn} for inference on ${\theta^\star}$. 

An interesting feature of the Gibbs posterior distribution is its dependence on a user-specified learning rate $\omega_n$.  Generally, the effect of the learning rate is to make the Gibbs posterior either more spread out (small learning rate) or more sharply peaked at ${\theta^\star}$ (large learning rate).  In other words, $\omega_n$ has a substantial effect on the variance of the Gibbs posterior distribution and the frequentist coverage probability of its credible sets for ${\theta^\star}$.  It is possible to simply set $\omega_n = 1$ and ignore the learning rate, but, it can be helpful to tune the learning rate for better finite-sample performance.  This is precisely the strategy laid out in \citet{syring.martin.scaling}, which we employ in our examples that follow in Sections~5.1 and 6.1.  Roughly speaking, their method iteratively updates the learning rate $\omega_n$ until $95\%$ Gibbs posterior credible sets for ${\theta^\star}$ have approximately $95\%$ coverage with respect to a bootstrap sample of estimated ${\theta^\star}$ values.  Other learning rate tuning procedures are available; see, for example, the recent review paper \citet{wu.martin.2020}.      

To complete our specification of the Gibbs posterior distribution for the Youden index cutoff we need a prior distribution $\Pi$ on $\Theta$.  For the multi-class setting a flexible choice of prior distribution for the cutoff ${\theta}$ is a $(k-1)$-dimensional ordered independent Normal distribution.  Specifically, let ${\eta} \sim \nm_{k-1}({\mu}, {\Sigma})$ where ${\mu}$ is a $(k-1)$-dimensional mean vector and ${\Sigma}$ is a $(k-1)\times (k-1)$ diagonal covariance matrix.  Then, define the prior distribution for ${\theta}$ by taking  ${\theta} = (\eta_{(1)}, \ldots, \eta_{(k-1)})^\top$, that is, ${\theta}$ has the distribution of the order statistics of ${\eta}$.  For example, if $k=3$ diagnostic categories, then ${\theta} = (\theta_1, \theta_2)^\top$ is a $2$-dimensional cutoff vector with the density function of a bivariate normal distribution for ${\eta}=(\eta_1, \eta_2)^\top$ restricted and normalized to the half-space $\{\eta_2\geq \eta_1\}\subset \mathbb{R}^2$.  This type of prior distribution enforces the ordering of ${\theta}$ while allowing for either vague or informative prior information by varying the hyperparameters $\mu$ and $\Sigma$.  We emphasize that the nonparametric Bayesian method introduced by \citet{carvalho.2018} models the cutoff indirectly as a functional of a mixture distribution and cannot directly incorporate prior information about the cutoff.

We define the Gibbs posterior distribution for the smooth, covariate-adjusted cutoff $\theta^\star(z)$ using \eqref{eq:gibbs} with the objective function in \eqref{eq:empirical_Rn_func} parametrized by the vector parameter ${\beta}\in\mathbb{R}^d$.  \citet{shen.ghosal} provides guidance on prior distributions that provide good large-sample performance.  Their hierarchical random series prior distributions are composed of a marginal prior distribution on the number of basis functions $D$ and a conditional prior on $\beta$ given $D=d$.  Allowing the number of basis functions to vary can improve the fit of the curve $\theta(z)$ in practice, but requires more advanced posterior sampling algorithms, like reversible jump MCMC.  We opt for a fixed number of basis functions $d$ so that our Gibbs posterior distribution can be sampled using the standard Metropolis-Hastings-within-Gibbs algorithm.  The general theory in \citet{shen.ghosal} suggests using independent normal or independent exponential prior distributions on $\beta_j$ for $j=1, \ldots, d$.

\section{Theoretical support}
\label{S:theory}
In this section we discuss the large-sample properties of the Gibbs posterior distribution for the cutoff in the multi-class setting described in Section~\ref{SS:YI} and in the covariate-adjusted setting described in Section~\ref{S:covariate}.  Proofs of Theorems~\ref{thm:1} and \ref{thm:2} are deferred to the Web Appendices.
\subsection{Multi-class setting}
Theorem~\ref{thm:1} below says that as long as the prior distribution places sufficient mass near ${\theta^\star}$, then the Gibbs posterior distribution concentrates in $P^n-$probability on the true cutoff $\theta^\star$ at a rate determined by the smoothness of the conditional CDFs $F_j$ evaluated at $\theta^\star_j$ for $j=1,\ldots, k-1$.
\begin{assump}\hspace{1cm}\\\vspace{-.5cm}
\label{assump:thm1}
\begin{enumerate}
    \item[i. ] Let $L(\eps):=\{\theta: \|\theta - {\theta^\star}\|< \eps\}$.  The prior distribution $\Pi$ has density $\pi$ strictly bounded away from zero on $L(\eps)$ for all sufficiently small $\eps>0$.
\item[ii. ] Define $R({\theta}):= E\left[\sum_{j=1}^{k-1}\ell(\theta_j;Y,X)\right]$.  There exist positive constants $\gamma>1/2$ and $\gamma\geq\eta>0$ such that for all small enough $\eps>0$ 
\begin{equation}
    \label{eq:ident1}
    \|{\theta}-{\theta^\star}\| > \eps \Longrightarrow   R({\theta}) - R({\theta^\star}) \gtrsim \eps^{\gamma};
\end{equation}
and, 
\begin{equation}
\label{eq:ident2}
\begin{aligned}
    \|{\theta}-{\theta^\star}\| < \eps \Longrightarrow&   \sum_{j=1}^{k-1}\biggl\{|F_{j+1}(\theta_j) - F_{j+1}(\theta^\star_j)|+ |F_{j}(\theta_j) - F_{j}(\theta^\star_j)|\biggr\} \lesssim \eps^{\eta}.
\end{aligned}
\end{equation}
\end{enumerate}
\end{assump}
\begin{theorem}
\label{thm:1}
Let $M_n > 0$ be a positive sequence and let $\eps_n$ be a vanishing sequence satisfying $M_n\eps_n\rightarrow 0$, $n^{-1/2}(M_n\eps_n)^{1/2-\gamma}\rightarrow 0$, and $n(M_n\eps_n)^\gamma\{\log M_n\eps_n\}^{-1}\rightarrow\infty$.  Define $A_n := \{\theta:\|\theta - {\theta^\star}\| > M_n\eps_n\}$.  If Assumption~\ref{assump:thm1} holds then
\begin{itemize}
    \item[a) ] for known $p_j$ $j=1, \ldots, k-1$ in (3); or,
    \item[b) ] for unknown $p_j$ estimated by sample proportions $\hat p_j$, and $\eps_n$ satisfying $\eps_n^\gamma \gtrsim n^{-1/2}$, 
\end{itemize} the Gibbs posterior probability of $A_n$ vanishes in $P^n-$probability as $n\rightarrow\infty$. \end{theorem}
Assumption~\ref{assump:thm1} ii. relates to the smoothness of the conditional CDFs  $F_j$ and $F_{j+1}$ in neighborhoods of $\theta^\star_j$.  In the regular case the conditional CDFs $F_j$, $j=1, \ldots, k-1$, are twice-differentiable, and because they are stochastically ordered the corresponding densities $f_j$ and $f_{j+1}$ intersect at only one value, which equals $\theta_j^\star$.  A Taylor expansion of $R(\theta)$ at ${\theta^\star}$ implies $\|\theta - {\theta^\star}\|_2^2\lesssim R(\theta)-R({\theta^\star}) \lesssim \|\theta - {\theta^\star}\|_2^2$, and the Taylor expansion along with the Cauchy-Schwarz inequality implies Assumption~\ref{assump:thm1} ii. holds with $\gamma = 2$ and $\eta = 1$. When every $p_j$ $j=1, \ldots, k-1$ is known Theorem~\ref{thm:1} holds for $M_n = \log(n)$ and $\eps_n=n^{-1/3}$ and when the probabilities $p_j$ are unknown the theorem holds for any diverging $M_n$, such as $\log\log n$, and $\eps_n = n^{-1/4}$. 

\subsection{Covariate-adjusted cutoff}
Theorem~\ref{thm:2} below establishes consistency of the Gibbs posterior distribution for inference on the smooth, covariate-adjusted cutoff function $\theta^\star(z)$ when $\theta^\star(z)$ is H\"older smooth with known exponent $\alpha$.    
\begin{assump}\hspace{1cm}\vspace{-5mm}\\
\label{assump:thm2}
\begin{enumerate}
\item[i.] For the same constant $H>0$ as in (2), for some constant $C>0$, and for all sufficiently small $\eps > 0$
\begin{equation}
\label{eq:prior.beta}
\Pi(\{\beta:\|\beta-{\beta'}\|_2 \leq \eps\}) \gtrsim e^{-Cd\log(1/\eps)},
\end{equation}
for all ${\beta'} \in \RR^d$ with $\|{\beta'}\|_\infty \leq H$;
\item[ii.] The diagnostic measure $X$ and patient covariate $Z$ have marginal densities $f$ on $\mathcal{X}$ and $g$ on $[0,1]$ bounded away from zero and $\infty$.
\item[iii.] The cutoff function $\theta^\star(z)$ is H\"older smooth with given exponent $\alpha$:
\[ |\theta^{\star ([\alpha])}(z) - \theta^{\star ([\alpha])}(z')| \leq L |z - z'|^{\alpha - [\alpha]}. \]
\end{enumerate}
\end{assump}
For a $d-$vector $\beta$ define $\|{\beta^\top B_d}\|:=\int_0^1 |{\beta^\top B_d(z)}|dz$.  Define ${\beta_d^\star}$ to be any $d-$vector satisfying \eqref{eq:approx.smooth}.
\begin{theorem}
\label{thm:2}
For any fixed $\eps>0$ define $A_n(d):=\{\beta\in\mathbb{R}^d: \|{\beta^\top B_d} - \theta^\star\| > \eps\}$.  If Assumption~\ref{assump:thm2} holds, then for every sufficiently large $d>0$ the Gibbs posterior probability $\Pi_n[A_n(d)]\rightarrow0$ in $P^n-$probability as $n\rightarrow\infty$.    
\end{theorem} 
\section{Examples}
\label{S:examples}
\subsection{Dementia diagnosis}
\label{SS:TMT}
We revisit the analysis of the three-class TMT data in \citet{carvalho.2018} using the proposed Gibbs posterior distribution for inference on the cutoff $\theta = (\theta_1, \theta_2)^\top$.  The data set contains $245$ total observations of time it took a participant to complete the TMT.  The data is split into three categories: $170$ unimpaired subjects, $52$ with mild cognitive impairment, and $23$ with dementia.  We randomly split the data into two sets, one for analysis and one to act as prior data to illustrate the impact of an informative prior distribution on inference for the cutoff ${\theta^\star}$.  Using the set serving as prior data we minimize the empirical objective function $R_n(\theta)$ and estimate the  cutoff to be ${\hat\theta} = (51.03,\,72.67)^\top$, which will be used as the mean of an informative ordered independent normal prior distribution for $\theta$.  For the prior standard deviations, we use $1000$ bootstrap resamples of the prior data to compute $95\%$ bootstrap percentile intervals.  Then we set the prior standard deviations to one fourth the width of $95\%$ bootstrap percentile interval estimates, $\frac{1}{4}(\hat\theta^B_{j,(975)}- \,\hat\theta^B_{j,(25)})$ for $j=1, 2$, using $1000$ bootstrap estimates $((\hat\theta^B_{1,1}, \hat\theta^B_{2,1}),\ldots,(\hat\theta^B_{1,1000}, \hat\theta^B_{2,1000}))$.  These informative prior standard deviations equal $(3.63, 6.81)$.  For comparison we also use a vague normal prior with standard deviations equal to $20$. 

We sample the Gibbs posterior for $\theta$ using the remaining half of the data.  The variance of the Gibbs posterior distribution is influenced by the learning rate $\omega_n$, and we use the GPC algorithm \citep{syring.martin.scaling} to choose a data-dependent learning rate $\omega_n$ that helps calibrate posterior credible intervals.  We compare the $95\%$ Gibbs posterior credible intervals for $\theta$ using both the informative and vague prior distribution to the bootstrap percentile intervals and the Dirichlet process mixture model of \citet{carvalho.2018}; these intervals are summarized in Table~\ref{tbl:TMT}.  Boxplots of the posterior samples and the bootstrapped M-estimates are displayed in Figure~\ref{fig:TMT}.  The informative prior distribution substantially concentrates the Gibbs posterior distribution compared with the Gibbs posterior that uses a vague prior, and provides the shortest interval estimates among the four methods.

\begin{figure}
    \centering
    \includegraphics[width = 0.45\textwidth]{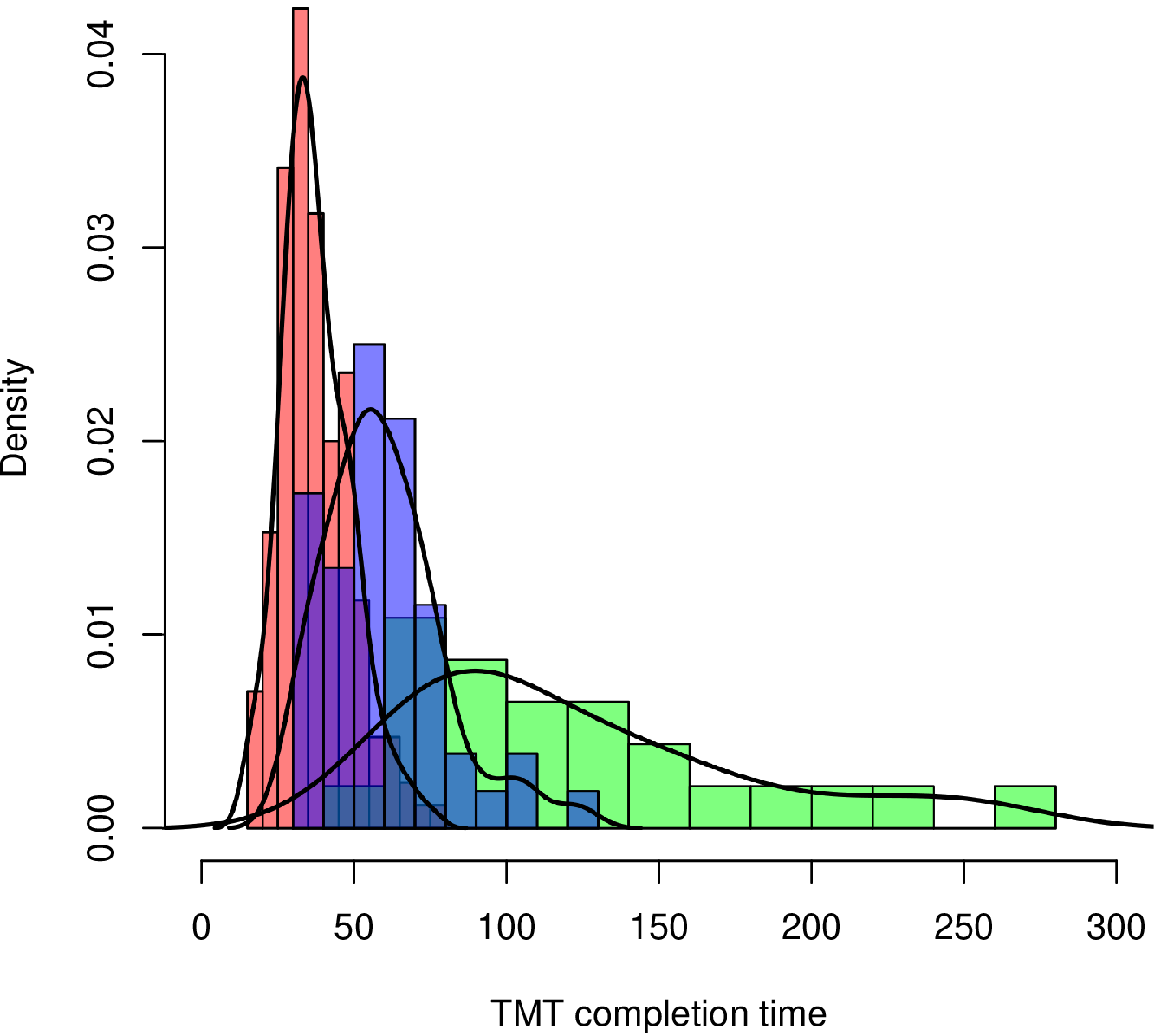}
      \includegraphics[width = 0.45\textwidth]{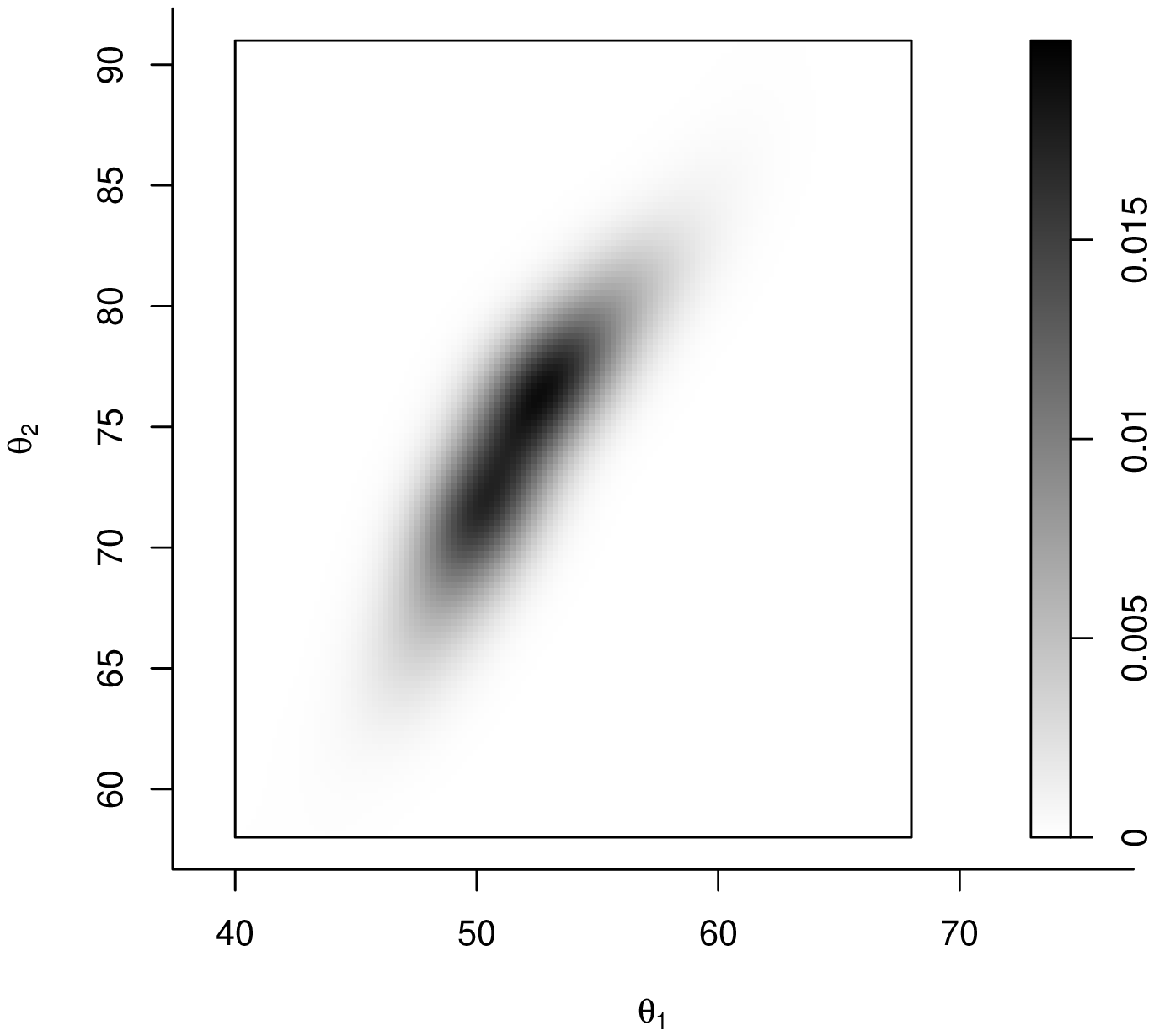}
    \includegraphics[width = 0.45\textwidth]{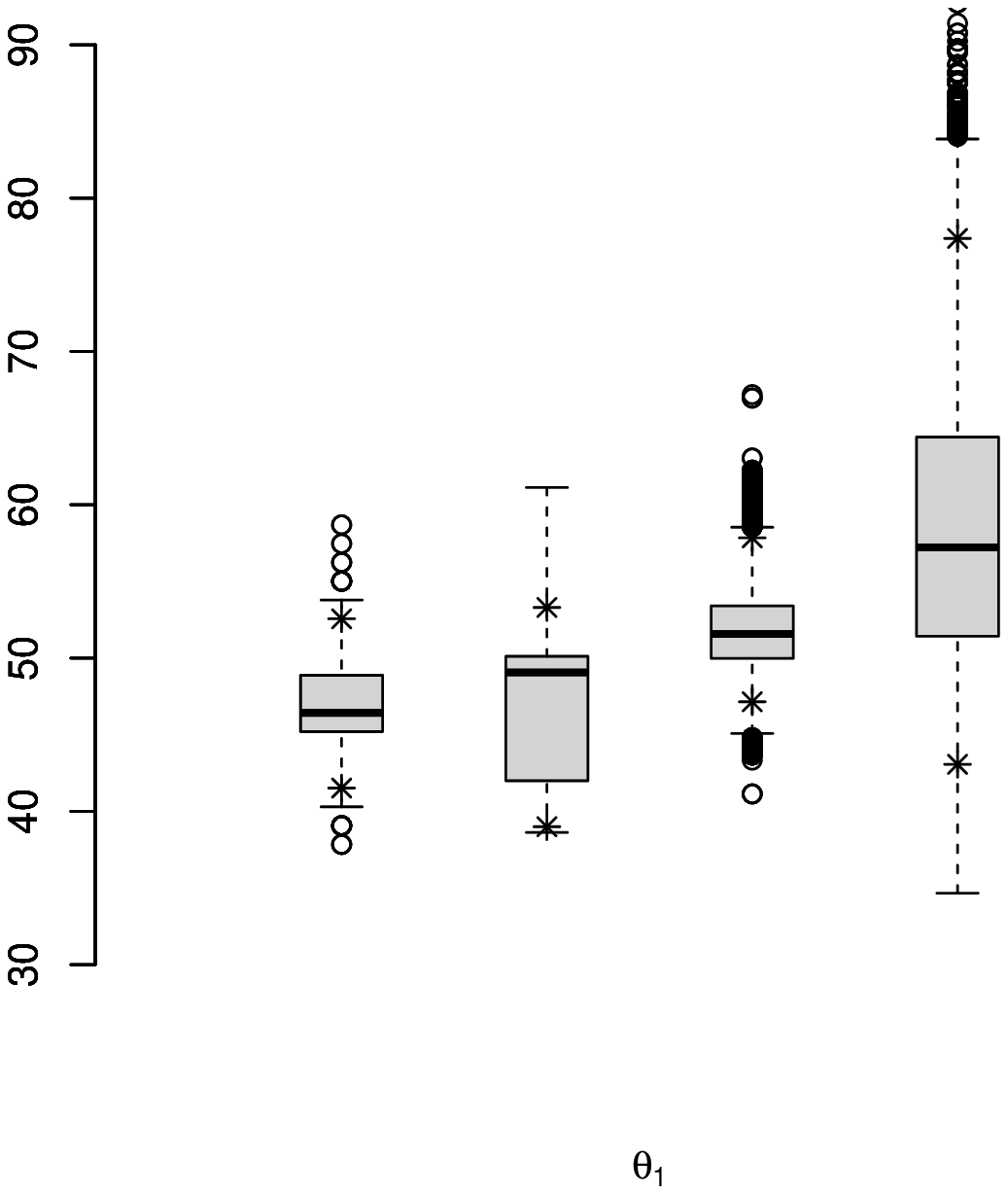}
    \includegraphics[width = 0.45\textwidth]{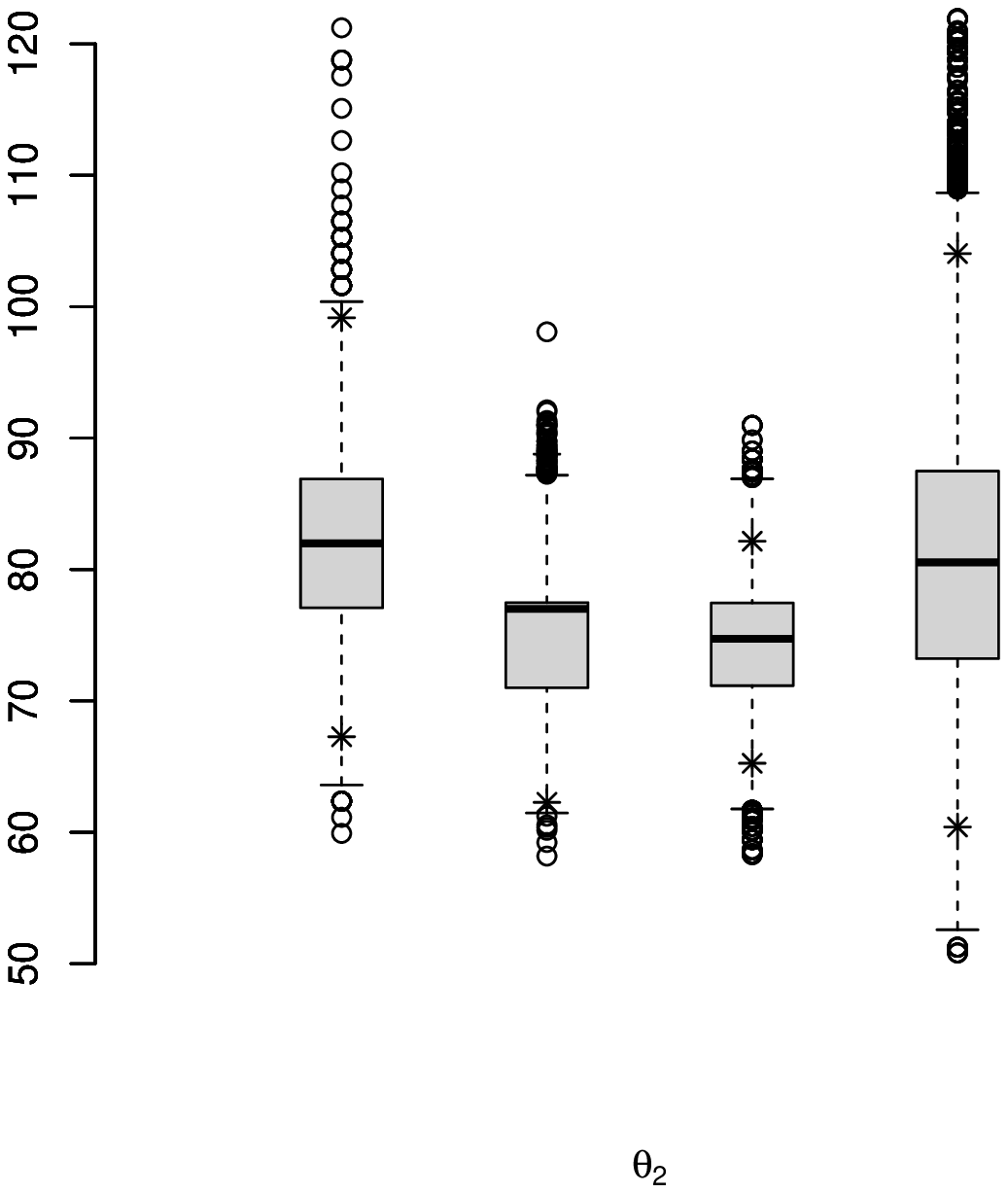}
    \caption{Top left: histograms and kernel density estimates of TMT completion times by diagnostic group for the TMT data analyzed in Section~5.1. Top right: heat map of Gibbs posterior density of the cutoff $\theta$ for the TMT data. Bottom left: boxplots of posterior samples and bootstrapped M-estimates of the first component of the cutoff $\theta_1$.  From left to right: Bayesian posterior samples, M-estimates, Gibbs posterior samples using informative prior distribution, and Gibbs posterior samples using vague prior distribution.  Asterisks indicate $95\%$ interval estimates.  Bottom right: same as bottom left for the second cutoff $\theta_2$.}
    \label{fig:TMT}
\end{figure}
\begin{table}[h]
\centering
\begin{tabular}{lcc}\toprule
                               & $\theta_1$       & $\theta_2$       \\ \toprule
\multicolumn{1}{l}{Bootstrap} & $(39.00, 53.50)$ & $(62.14, 89.40)$ \\
\multicolumn{1}{l}{Bayes}     & $(41.52, 52.56)$ & $(67.27, 99.15)$ \\
\multicolumn{1}{l}{Gibbs GPC*} & $(47.28, 58.33)$ & $(62.26, 87.74)$\\
\multicolumn{1}{l}{Gibbs GPC} & $(43.07,77.37)$ & $(60.40, 104.04)$\\
\bottomrule
\end{tabular}
\caption{$95\%$ interval estimates for the optimal TMT completion time cutoff $\theta$ using the percentile bootstrap intervals for the M-estimator, the Bayesian posterior, the Gibbs posterior based on an informative prior distribution (indicated by $*$), and the Gibbs posterior based on a vague prior distribution.}
\label{tbl:TMT}
\end{table}
\subsection{Covariate-adjusted diabetes diagnosis}
\label{SS:diab}
In this section we apply the Gibbs posterior distribution for a smooth cutoff function to a diabetes diagnostic data set studied by \citet{carvalho.2017} who used a nonparametric Bayesian regression model to estimate the cutoff.  Of $286$ patients, blood glucose measurements were used to classify $198$ as nondiabetic and $88$ as diabetic.  Our interest is in assessing how those diagnoses vary by age by estimating the age-adjusted cutoff function.  In \citet{carvalho.2017}, the authors find that a cubic b-spline with $d=4$ knots (no interior knots) worked best.  We compare their method to a Gibbs posterior with $8$ knots fixed at ``ages" $-20,\,-10,\,0,\,20,\,89,\,110,\,120,\,130$ since the minimum and maximum ages in the data were $20$ and $89$.  We used flat priors on $\beta_j$ for $j=1,...,4$, and a constant learning rate $\omega_n = 1$.  

The data, posterior mean functions, and $95\%$ credible bands are displayed in Figure~\ref{fig:gibbs_and_bayes_diabetes}.  The Gibbs posterior mean and Bayesian posterior mean of $\theta(z)$ are very similar and both models suggest (advanced) age increases the blood glucose reading at which diabetes is diagnosed.  The Gibbs posterior mean function is roughly quadratic while the Bayesian posterior mean function is closer to a piecewise linear function, flat from age 20 to 50 and then increasing with age.  The slight difference in shapes can be explained in part by the high-leverage observations at ages 27 and 28 where there are two diabetic patients with blood glucose levels of 139 and 330.  The objective function in \eqref{eq:empirical_Rn_func} is indifferent to $\theta(z)$ curves passing between these two points, which explains the high variation and upward slope of the Gibbs posterior for $\theta(z)$ in the age range $20-35$.  It is worth pointing out that the standard logistic regression model also suggests a quadratic relationship between age and diabetes diagnosis. Age is not a significant predictor in the logistic regression of diabetes diagnosis on age and blood glucose level, but becomes significant if age$^2$ is also included in the model.          
\begin{figure}
	\centering
		\includegraphics[width=5in]{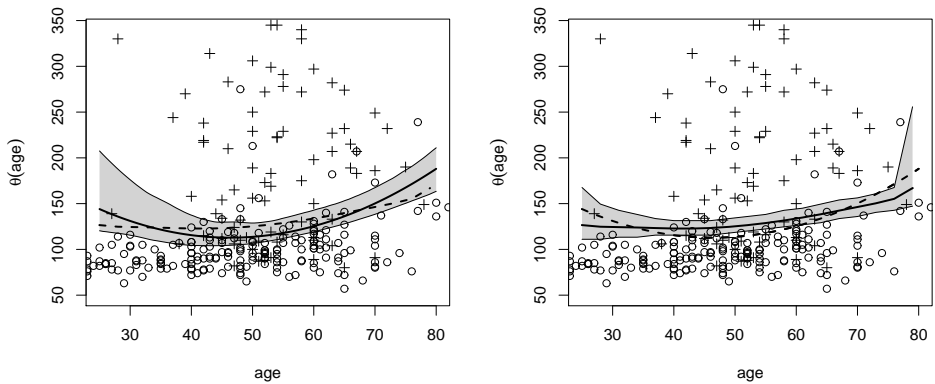}
			\caption{Blood glucose levels of patients diagnosed diabetic ($+$) and nondiabetic ($\circ$) along with posterior mean cutoff functions of age and $95\%$ credible bands for the Gibbs posterior (left) and Bayes posterior (right) in solid curves. Dashed curves show Bayes posterior mean (left) and Gibbs posterior mean (right) for ease of comparison.}
	\label{fig:gibbs_and_bayes_diabetes}
\end{figure}

\section{Simulations}
\label{ss:examples}
\subsection{Multi-class cutoff}
In addition to the data analyses in Section~5 we consider simulations for evaluating the coverage probability and length properties of Gibbs posterior credible intervals compared to Bayesian and bootstrap-based alternatives.  Examples 1 -- 3 detailed below correspond to scenarios 2 -- 4 in \citet{carvalho.2018}.  In all three examples there are $k=3$ equally-sampled diagnostic categories with $n$ observations per category, similar to a case-control study.  
\begin{enumerate}
	\item[1.] $F_j(x)$ is the cumulative distribution function of $\gam(2,1)$, $\gam(3,1)$, and $\gam(5,2)$ for $j=1,2,3$ where $\gam(\alpha,\beta)$ denotes the Gamma distribution with shape and scale parameters $\alpha$ and $\beta$.
	\item[2.] $F_1(x)$ is the cumulative distribution function of the normal mixture $\frac{1}{2}\nm(-1.5,0.5^2) + \frac{1}{2}\nm(0.5,1)$; $F_2(x)$ is the cumulative distribution function of the normal mixture $\frac{1}{2}\nm(1,1) + \frac{1}{2}\nm(4,1.5^2)$; and $F_3(x)$ is the cumulative distribution function $\nm(5,2^2)$, where $\nm(\mu, \sigma^2)$ denotes the normal distribution with mean $\mu$ and variance $\sigma^2$.
	\item[3.] $F_1(x)$ is the Student's t cumulative distribution function with $2$ degrees of freedom; $F_2(x)$ is the cumulative distribution function of a Beta distribution with shape and scale parameters equal to $2$; and $F_3(x)$ is the cumulative distribution function of a Chi-Squared distribution with one degree of freedom.
\end{enumerate}
For each simulation setting we sampled $500$ sets of data for sample sizes $n = 50$ and $n = 200$ observations in each diagnostic category and recorded the lengths and coverage probabilities of $95\%$ confidence/credible interval estimates for the cutoff.  We compared three methods of inference: percentile bootstrap, the non-parametric Bayes procedure from \citet{carvalho.2018}, and the Gibbs posterior distribution with learning rate $\omega_n$ determined by the GPC calibration procedure from \citet{syring.martin.scaling}.  We used both informative normal priors and vague normal priors for the Gibbs posterior distribution.  The informative priors are centered at $\theta^\star$ and use standard deviation $1/2$ in examples 1 and 2, and standard deviation $1/4$ in example 3.  

The results of the simulation study are summarized in Table~\ref{tbl:sim}.  The main takeaway from the simulations is that while all three inference methods produce interval estimates that come close to achieving their nominal coverage probabilities, the introduction of accurate prior information can lead to substantial improvements in precision, especially at smaller sample sizes.  Since only the Gibbs posterior method can incorporate an informative prior distribution, it should be preferred when such a prior distribution is available.   

\begin{table}
\resizebox{\textwidth}{!}{\begin{tabular}{lcccc}
\toprule
& \multicolumn{2}{c}{n = 50} & \multicolumn{2}{c}{n = 200}\\ 
& \multicolumn{1}{l}{Average length} & \multicolumn{1}{l}{Coverage proportion (\%)} & \multicolumn{1}{l}{Average length} & \multicolumn{1}{l}{Coverage proportion (\%)} \\ 
\bottomrule
\multicolumn{1}{l}{Bootstrap} & 1.57, 1.78                         & \multicolumn{1}{c}{92, 90}                  & 1.08, 1.19                         & 94, 92                                       \\
\multicolumn{1}{l}{Bayes}     & 1.89, 1.79                         & \multicolumn{1}{c}{95, 95}                  & 1.15, 1.02                         & 92, 92                                       \\
\multicolumn{1}{l}{Gibbs GPC*} & 1.16, 1.39                         & \multicolumn{1}{c}{94, 99}                  & 0.81, 0.95                         & 90, 94                                       \\
\multicolumn{1}{l}{Gibbs GPC} & 1.90, 2.18                         & \multicolumn{1}{c}{90, 95}                  & 1.01, 1.18                         & 88, 91                                       \\\bottomrule
& \multicolumn{1}{l}{}               & \multicolumn{1}{l}{}                         & \multicolumn{1}{l}{}               & \multicolumn{1}{l}{}                         \\
\toprule
& \multicolumn{2}{c}{n = 50}                & \multicolumn{2}{c}{n = 200}                                                       \\ 
& \multicolumn{1}{l}{Average length} & \multicolumn{1}{l}{Coverage proportion (\%)} & \multicolumn{1}{l}{Average length} & \multicolumn{1}{l}{Coverage proportion (\%)} \\ 
\bottomrule
\multicolumn{1}{l}{Bootstrap} & 1.32, 1.40                         & \multicolumn{1}{c}{91, 90}                  & 1.06, 1.19                                   &   94, 93\\
\multicolumn{1}{l}{Bayes}     & 1.70, 1.92                         & \multicolumn{1}{c}{94, 92}                  &  1.20, 1.22                                  &    95, 94\\
\multicolumn{1}{l}{Gibbs GPC*} & 1.27, 1.38                         & \multicolumn{1}{c}{93, 92}                  &   0.99, 1.14                                 &     90, 92  \\  
\multicolumn{1}{l}{Gibbs GPC} & 1.97, 2.45                         & \multicolumn{1}{c}{93, 94}                  &   1.28, 1.59                                 &     95, 90  \\ \bottomrule & \multicolumn{1}{l}{}               & \multicolumn{1}{l}{}                         & \multicolumn{1}{l}{}               & \multicolumn{1}{l}{}                         \\
\toprule
& \multicolumn{2}{c}{n = 50}                                                        & \multicolumn{2}{c}{n = 200}                                                       \\ 
& \multicolumn{1}{l}{Average length} & \multicolumn{1}{l}{Coverage proportion (\%)} & \multicolumn{1}{l}{Average length} & \multicolumn{1}{l}{Coverage proportion (\%)} \\ 
\bottomrule 
\multicolumn{1}{l}{Bootstrap} &  0.21, 0.23                       & \multicolumn{1}{c}{90, 90}                  &   0.13, 0.11                              &      93, 91\\
\multicolumn{1}{l}{Bayes}     &   0.18, 0.18                     & \multicolumn{1}{c}{96, 88}                  &              0.09, 0.09                      &     84, 59\\
\multicolumn{1}{l}{Gibbs GPC*} &     0.26, 0.29                   & \multicolumn{1}{c}{98, 97}                  &     0.13, 0.13                              &      92, 93    \\
\multicolumn{1}{l}{Gibbs GPC} &     0.31, 0.36                  & \multicolumn{1}{c}{97, 97}                  &     0.13, 0.14                              &      93, 94     \\\bottomrule
\end{tabular}}
\caption{Average lengths and coverage proportions of interval estimates for Youden index cutoffs $\theta_1$ and $\theta_2$ for simulations 1. (top), 2. (middle), and 3. (bottom) in Section~6.1.}
\label{tbl:sim}
\end{table}
\subsection{Covariate-adjusted cutoff}

We investigated the performance of the Gibbs posterior for the covariate-adjusted cutoff function in three simulation examples.  The first two examples were taken from \citet{carvalho.2017} and the third is closely related.  Example 1 is similar to a simple linear regression model; example 2 is a linear model with heteroscedasticity; and, example 3 is a heavy-tailed regression model. 

\begin{enumerate}
	\item[1.] The conditional CDFs of the diagnostic measure given covariate value are $F_1(x) = \Phi(0.5+z, 1.5)$ and $F_2(x) = \Phi(2+4z, 2)$ where $\Phi(\mu,\sigma)$ denotes the CDF of a normal random variable with mean $\mu$ and standard deviation $\sigma$.  

\item[2.] $F_1(x) = \Phi(3+1.5\sin(\pi z), 0.2+\exp(z))$ and $F_2(x) = \Phi(5+1.5z+1.5\sin(z),(1.5+\Phi(10z-2,1))^{1/2})$.

\item[3.] $F_1(x) = T(3+1.5\sin(\pi z), 2)$ and $F_2(x) = T(5+1.5z+1.5\sin(z), 2)$ where $T(m,\tau)$ denotes the CDF of a Student t random variable with location parameter $m$ and degrees of freedom $\tau$.  
\end{enumerate}

Observations were simulated by randomly sampling covariates $z_i\sim\unif(0,1)$ from a standard uniform distribution and, conditionally, simulating diagnostic measurement $x_i$ from the each of the above models.  For Example~1, we simulated $n=100$ samples for each diagnostic group, while for Examples~2 and 3 we simulated $n=200$ samples per group.  In each example we compare the Gibbs posterior mean function to the Bayesian posterior mean function using the method of \citet{carvalho.2017}.  Both methods recover the true cutoff function on average over $100$ simulated data sets.  The Gibbs posterior, which uses the default learning rate $\omega_n=1$, exhibits slightly more variation in posterior mean than the Bayesian posterior; see Figure~\ref{fig:reg}.

\begin{figure}
	\centering
		\includegraphics[width=5in]{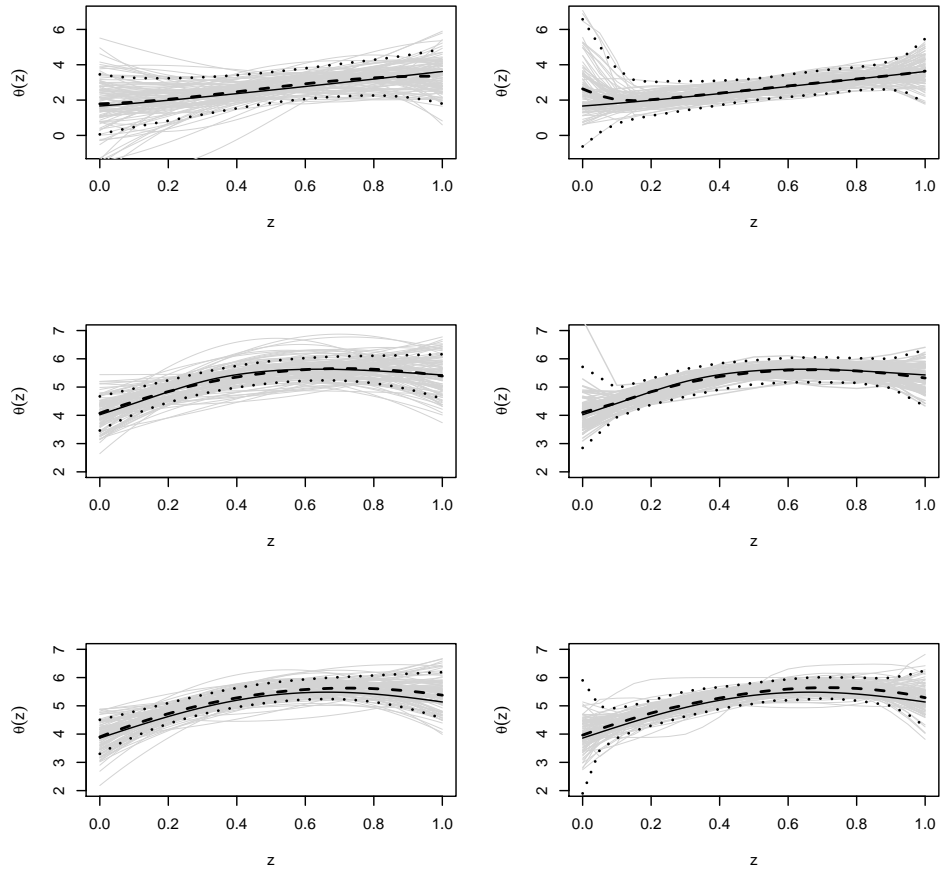}
	\caption{Simulations 1 (top), 2 (middle), and 3 (bottom) form Section~6.2 for the Gibbs (left column) and Bayes (right column) approaches.  Posterior means for $100$ simulations in gray, mean over all runs in dashed line, mean of $2.5^{th}$ and $97.5^{th}$ posterior quantiles in dotted lines and true cutoff function in solid black line.}
	\label{fig:reg}
\end{figure}

\section{Discussion}
Inference on Youden's index cutoff, the index itself, and related ROC curves has been an important problem in medical statistics in recent years; see \citet{nakas.2017}.  Our contribution addresses the challenging problems of robust modeling and incorporating prior information about Youden's index cutoff.  Current techniques address only one or the other concern, opting either for likelihood-based methods that can incorporate prior information via a Bayesian posterior distribution, or using a nonparametric approach that is robust to data distributions.  The proposed Gibbs posterior approach to inference on the cutoff offers a promising alternative to parametric models and M-estimation.  The Gibbs posterior is robust to data distributions and offers simple informative prior specification and integration.  It seamlessly incorporates covariate information and can be applied to data with any number of diagnostic categories.  The Gibbs posterior has favorable large-sample properties, and can be tuned to deliver valid credible sets for the cutoff in practice.

There are limitations to using the Gibbs posterior distribution for inference.  Our simulations and real-data examples suggest it may be inefficient compared to the nonparametric Bayesian method described in \citet{carvalho.2018} and the bootstrap when only vague prior information about the cutoff is available. In such cases there is no reason to prefer the Gibbs posterior to the bootstrap.   

An important problem we have not covered here is how to implement Gibbs posterior inference for the multi-class cutoff in the presence of covariates. The additional challenge brought by having multiple classes is the requirement the regression functions for each cutoff are ordered.  The same challenge arises in quantile regression when more than one quantile is modeled simultaneously. So, techniques from the literature on quantile regression could shed light on how to solve the multi-class problem with covariates.

\appendix

\section{Proof of Theorem 4.1}

\subsection{Preliminary results}
\label{proof:gibbs_rate}

Towards a proof of Theorem~\ref{thm:1} write the Gibbs posterior probability of $A_n$ as
\[\Pi_n(A_n) = \frac{N_n(A_n)}{D_n} = \frac{\int_{A_n} e^{-\omega n [R_n(\theta) - R_n({\theta^\star})]}\Pi(d\theta)}{\int_\Theta e^{-\omega n [R_n(\theta) - R_n({\theta^\star})]}\Pi(d\theta)}.\]
The following preliminary results will be used to control the numerator $N_n(A_n)$ (Lemmas~\ref{lem:bracket_bound}-\ref{lem:num}) and the denominator $D_n$ (Lemma~\ref{lem:den}) in $P^n-$probability.  

Consider the family of functions defined by the loss differences $\mathcal{L}_\delta := \{\ell_{ \theta} - \ell_{{\theta^\star}}: \|\theta-{\theta^\star}\|<\delta\}$.  The following lemma bounds the complexity of this family by its bracketing number in $L_2(P)$.
\begin{lem}
\label{lem:bracket_bound}
$N_{[]}(\eps, \mathcal{L_\delta}, L_2(P)) \lesssim \delta / \eps^2$.
\end{lem}
\begin{proof}
Define a grid $-\infty = t_0 < t_2 < \cdots < t_J = \infty$ such that there exists a subsequence of integers $s_1,\ldots, s_k$ where $t_{s_j} = \theta^\star_j$. Define $s_0 = 0$.  And, let the grid satisfy $P(X\in (t_j, t_{j+1})) < \eps^2\min_j p_j^2$ for a fixed $\eps>0$.  For $\ell = 1, \ldots, k-1$, define the functional brackets $[\tfrac{1}{\min_j p_j}1\{x \in [t_j, \theta_\ell^\star]\}, \, \tfrac{1}{\min_j p_j}1\{x \in (t_{j-1}, \theta_\ell^\star]\}]$ for $s_{\ell-1}<j\leq s_\ell$.  These are brackets of $\mathcal{L}_\delta$ with $L_2(P)-$size $\eps$ and their total number can be taken less than $2k\delta/(\eps^2\min_j p_j^2)$.  
\end{proof}
Lemma~\ref{lem:bracket_bound} implies the following bound on the entropy with bracketing of $\mathcal{L}_\delta$,
\begin{align}
\label{eq:entropy}
    J_{[]}(\delta^{1/2}, \mathcal{L}_\delta, L_2(P)) &= \int_0^{\delta^{1/2}} \sqrt{\log N_{[]}(\eps, \mathcal{L_\delta}, L_2(P))}d\eps \lesssim \delta^{1/2}.
\end{align}
Denote the empirical process $\mathbb{G}_n(\ell_{\theta} - \ell_{{\theta^\star}}):=\sqrt{n}[\mathbb{P}_n(\ell_{\theta} - \ell_{{\theta^\star}}) - E(\ell_{\theta} - \ell_{{\theta^\star}})]$ where $\mathbb{P}_n(\cdot)$ denotes expectation with respect to empirical measure.  The function $F(\theta) = \frac{1}{\min jp_j}1\{\|\theta - {\theta^\star}\|<\delta\}$ is an envelope for $\mathcal{L}_\delta$ with $L_2(P)-$size $\|F\|_{L_2(P)}\lesssim \delta^{1/2}$.  Then, \eqref{eq:entropy} along with Corollary 19.35 in \citet{vaart} implies the following maximal inequality:
\begin{lem}
\label{lem:maximal}
$E\{\sup_{\|\theta-{\theta^\star}\|<\delta} |\mathbb{G}_n(\ell_{\theta} - \ell_{{\theta^\star}})|\}\lesssim \delta^{1/2}.$
\end{lem}
The maximal inequality provided by Lemma~\ref{lem:maximal} along with the bound on $R(\theta)-R(\theta^\star)$ provided by Assumption~\ref{assump:thm1} ii. can be used to prove the following uniform probability bound. 
\begin{lem}
\label{lem:num}
There exists $K > 0$ such that 
\[ P\Bigl(\sup_{\|\theta-{\theta^\star}\| > M_n\eps_n} \{R_n({\theta^\star}) - R_n(\theta) \} > -K (M_n\eps_n)^{\gamma} \Bigr) \to 0, \quad \text{as $n \to \infty$}. \]
\end{lem}
\begin{proof}
Start with the identity 
\[ R_n({\theta^\star}) - R_n(\theta) = \{ R({\theta^\star}) - R(\theta) \} - n^{-1/2} \GG_n(\ell_{\theta} - \ell_{{\theta^\star}}).  \]
Next, since the supremum of a sum is no more than the sum of the suprema, we get 
\[ \sup_{\|\theta-{\theta^\star}\| > \eps} \{ R_n({\theta^\star}) - R_n(\theta) \} \leq \sup_{\|\theta-{\theta^\star}\| > \eps} \{ R({\theta^\star}) - R(\theta) \} + n^{-1/2} \sup_{\|\theta-{\theta^\star}\| > \eps} |\GG_n(\ell_{\theta} - \ell_{{\theta^\star}})|; \]
the second inequality comes from putting absolute value on the empirical process term.  From Assumption~\ref{assump:thm1} ii., we get 
\[ \sup_{\|{\theta}-{\theta^\star}\| > \eps} \{ R_n({\theta^\star}) - R_n(\theta) \} \leq -C \eps^{\gamma} + n^{-1/2} \sup_{\|{\theta}-{\theta^\star}\| > \eps} |\GG_n(\ell_{\theta} - \ell_{{\theta^\star}})|. \]
Now, following the proof of Theorem~5.52 from \citet{vaart} or of Theorem~1 in \citet{wong.shen.1995}, introduce ``shells'' $\{\theta: 2^m \eps < \|\theta-{\theta^\star}\| \leq 2^{m+1} \eps\}$ for integers $m$.  On these shells, we can use both the bound in Assumption~\ref{assump:thm1} ii. and the maximal inequality in Lemma~\ref{lem:maximal}.  That is, 
\begin{align*}
\sup_{\|\theta-{\theta^\star}\| > M_n\eps_n} & \{R_n({\theta^\star}) - R_n(\theta)\} > -K (M_n\eps_n)^{\gamma} \\
& \implies \sup_{2^m M_n\eps_n < \|\theta-{\theta^\star}\| \leq 2^{m+1} M_n\eps_n} \{R_n({\theta^\star}) - R_n(\theta)\} > -K (M_n\eps_n)^{\gamma} \quad \exists \; m \geq 0 \\
& \implies n^{-1/2} \sup_{2^m M_n\eps_n < \|\theta-{\theta^\star}\| <  2^{m+1} M_n\eps_n} |\GG_n(\ell_{\theta}-\ell_{{\theta^\star}})| \geq C (2^m M_n\eps_n)^{\gamma} - K (M_n\eps_n)^{\gamma} \\
& \implies n^{-1/2} \sup_{\|\theta-{\theta^\star}\| \leq  2^{m+1} M_n\eps_n} |\GG_n(\ell_{\theta}-\ell_{{\theta^\star}})| \geq C (2^m M_n\eps_n)^{\gamma} - K (M_n\eps_n)^{\gamma},
\end{align*}
If $K \leq C/2$, then $C (2^m M_n\eps_n)^{\gamma} - K (M_n\eps_n)^{\gamma} \geq C(2^m M_n\eps_n)^{\gamma} / 2$ for all $m \geq 0$.   
\begin{align*}
P\Bigl( \sup_{\|\theta-{\theta^\star}\| > M_n\eps_n} & \{R_n({\theta^\star}) - R_n(\theta)\} > -K (M_n\eps_n)^{\gamma} \Bigr) \\
& \leq \sum_{m \geq 0} P\Bigl(  n^{-1/2} \sup_{\|\theta-{\theta^\star}\| <  2^{m+1} M_n\eps_n} |\GG_n(\ell_{\theta}-\ell_{{\theta^\star}})| \geq C (2^m M_n\eps_n)^{\gamma} / 2 \Bigr)
\end{align*}
To the summands, apply Markov's inequality and Lemma~\ref{lem:maximal} to get 
\[ P\Bigl(  n^{-1/2} \sup_{\|\theta-{\theta^\star}\| <  2^{m+1} M_n\eps_n} |\GG_n(\ell_{\theta}-\ell_{{\theta^\star}})| \geq C (2^m M_n\eps_n)^{\gamma} / 2 \Bigr) \leq \frac{C' (2^{m+1} M_n\eps_n)^{1 / 2}}{n^{1/2} (2^m M_n\eps_n)^{\gamma}}. \]
Collecting the $n-$ and $m-$dependent terms on the right hand side of the above expression and simplifying we get the multiplicative factors $n^{-1/2}(M_n\eps_n)^{1/2-\gamma}$ and $\sum_{m=1}^\infty 2^{m(1/2-\gamma)}$.  Since $\gamma>1/2$ the sum converges while the first factor vanishes by assumption, and, consequently, the upper bound vanishes as $n\rightarrow\infty$.  
\end{proof}
\begin{lem}
\label{lem:den}
For any $Q>0$ and any sequence $s_n$ satisfying $s_n\rightarrow 0$ and $ns_n^\eta\rightarrow\infty$ \[D_n \gtrsim (s_n)^{k-1}e^{-Qn\omega s_n^\eta}\]
with $P^n-$probability converging to $1$.
\end{lem}
\begin{proof}
Define the sets $G_n:=\{\theta: R(\theta)-R(\theta^\star) \vee V(\ell_\theta - \ell_{\theta^\star}) \leq s_n^\eta\}$.  Lemma~1 in \citet{syring.martin.gibbs} shows 
\[ D_n \geq \Pi(G_n)e^{-2nQ\omega s_n^\eta}\]
with $P^n-$probability converging to $1$ for an arbitrary constant $Q>0$.

The next step is to quantify $\Pi(G_n)$ in terms of the $L_1$ metric on $\Theta$.  For one observation, consider the loss difference
\begin{align*}
&\ell(\theta;Y,X)-\ell({\theta^\star};Y,X)= \sum_{j=1} \ell(\theta_j,Y,X) - \ell(\theta_j^\star,Y,X).
\end{align*}
This difference has expectation 
\begin{align*}
E(&\ell(\theta;Y,X)-\ell({\theta^\star};Y,X))= \sum_{j=1}^{k-1}\ F_{j+1}(\theta_j)-F_{j}(\theta_j)-F_{j+1}(\theta_j^\star)+F_{j}(\theta_j^\star),
\end{align*}
and variance (bounded by the second moment), and using the facts $\theta_j<\theta_{j+1}$ and $\theta^\star_j < \theta^\star_{j+1}$,
\begin{align*}
&V(\ell(\theta;Y,X)-\ell({\theta^\star};Y,X))< \sum_{j=1}^{k-1} \frac{|F_{j+1}(\theta_j)-F_{j+1}(\theta^\star_j)|}{p_{j+1}}+\frac{|F_j(\theta_j)-F_j(\theta_j^\star)|}{p_j}.
\end{align*}

Next, by Assumption~\ref{assump:thm1} ii. and the triangle inequality, we can bound the above expectation and variance by the $L_1-$norm.  For any $\theta$ such that $\|\theta-\theta^\star\|<\eps$
\[\{R(\theta)-R(\theta^\star) \vee V(\ell_{\theta} - \ell_{{\theta^\star}})\} \lesssim \eps^\eta. \]
 For some appropriately chosen constant $C>0$ the sets $G_n$ contain the  $L_1-$neighborhoods
\[G_n \supset \{\theta: \|\theta-{\theta^\star}\|\leq Cs_n\}.\]

Finally, by the above arguments and Assumption~\ref{assump:thm1} i., we see that $\Pi(G_n)$ can be bounded
\begin{align*}
\Pi(G_n) &\geq \Pi(\{\theta: \|\theta-{\theta^\star}\|\leq Cs_n\})\gtrsim (s_n)^{k-1}.
\end{align*}
\end{proof}

\setcounter{section}{1}
\subsection{Proof of Theorem 4.1 a)}
\begin{proof}

First, consider the numerator $N_n(A_n)$ of the Gibbs posterior probability $\Pi_n(A_n)$
\begin{align*}
N_n(A_n) & = \int_{\|\theta-{\theta^\star}\| > M_n\eps_n} e^{-\omega n\{R_n(\theta) - R_n({\theta^\star})\}} \, \Pi(d\theta). \end{align*}
By Lemma~\ref{lem:num}, 
\[N_n(A_n)\lesssim e^{-nK\omega (M_n\eps_n)^\gamma}\]
with $P^n-$probability converging to $1$.  

Now, consider the denominator $D_n$ of the Gibbs posterior probability $\Pi_n(A_n)$. Apply Lemma~\ref{lem:den} with the choice $s_n = (M_n\eps_n)^{\gamma/\eta}$, noting these choices satisfy $s_n\rightarrow 0$ and $ns_n^{\eta}\rightarrow \infty$, and obtain the following in-probability lower bound 
\[ D_n \gtrsim (M_n\eps_n)^{(k-1)\gamma/\eta} e^{-Qn\omega(M_n\eps_n)^\gamma}. \]

With these bounds on $N_n(A_n)$ and $D_n$ in mind, and the fact that $\Pi_n(A_n)\leq 1$, bound $\Pi_n(A_n)$ as follows:
\begin{align*}
    \Pi_n(A_n) &\leq \frac{N_n(A_n)}{D_n}1\{D_n \geq (M_n\eps_n)^{(k-1)\gamma/\eta} e^{-Qn\omega(M_n\eps_n)^\gamma}\}1\{N_n(A_n)\leq e^{-nK\omega (M_n\eps_n)^\gamma}\}\\
    & + 1\{D_n < (M_n\eps_n)^{(k-1)\gamma/\eta} e^{-Qn\omega(M_n\eps_n)^\gamma}\} + 1\{N_n(A_n)> e^{-nK\omega (M_n\eps_n)^\gamma}\}\\
    &\lesssim e^{-\omega n(M_n\eps_n)^\gamma( K-Q)-(k-1)(\gamma/\eta)\log M_n\eps_n}  + 1\{D_n < (M_n\eps_n)^{(k-1)\gamma/\eta} e^{-Qn\omega(M_n\eps_n)^\gamma}\}\\& + 1\{N_n(A_n)> e^{-nK\omega(M_n\eps_n)^\gamma}\}.
\end{align*}
Take expectation of both sides to see that
\[E[\Pi_n(A_n)] \lesssim e^{-\omega n(M_n\eps_n)^\gamma( K-Q)-(k-1)(\gamma/\eta)\log M_n\eps_n} + o(1).\]
Since $Q$ is arbitrary $K-Q>0$ and $n(M_n\eps_n)^\gamma\{\log M_n\eps_n\}^{-1}\rightarrow \infty$ by assumption, the upper bound in the above display vanishes as $n\rightarrow \infty$. It follows by Markov's inequality that $P^n[\Pi_n(A_n)>\eps]\rightarrow 0$ as $n\rightarrow \infty$ for any $\eps>0$, completing the proof.  

\end{proof}

\subsection{Proof of Theorem 4.1 b)}
\label{ss:part.b}
\begin{proof}
The sample proportions $\hat p_j$ for $j=1, \ldots, k-1$ converge, individually, to $p_j$ by the LLN.  And, since $k$ is finite the vector $\hat p = (\hat p_1, \ldots, \hat p_{k-1})$ converges to the vector $p = (p_1, \ldots, p_{k-1})$ uniformly.  Further, Chebyshev's inequality implies $\hat p$ converges at rate $m_nn^{-1/2}$ where $m_n$ is any diverging sequence.  Therefore, denoting $(Y^n, X^n):=((Y_1, X_1),\ldots, (Y_n, X_n))$ the set $W:=\{(Y^n, X^n): \|\hat p - p\|_\infty > m_nn^{-1/2}\}$ has vanishing probability.  Since the posterior probability $\Pi_n(A_n)\leq 1$ it follows that
\begin{align*}
    E[\Pi_n(A_n)] &= E[\Pi_n(A_n)1(W)] + E[\Pi_n(A_n)1(W^c)]\\
    &= o(1) + E[\Pi_n(A_n)1(W^c)] 
\end{align*}
so we focus on bounding $E[\Pi_n(A_n)1(W^c)]$.

Let $R_n(\theta, {\hat p})$ denote the version of the empirical risk function with $p_j$ replaced by $\hat p_j$ for $j=1, \ldots, n$.  That is,
\[R_n(\theta, {\hat p}) = \frac{1}{n}\sum_{i=1}^n \sum_{j=1}^{k-1}\frac{1(x_i\leq \theta_j, y_i=j+1)}{ \hat p_{j+1}}-\frac{1(x_i\leq \theta_j, y_i=j)}{ \hat p_j} .\]
On $W^c$, $R_n(\theta, \hat p)$ is bounded above by
\[\frac{1}{n}\sum_{i=1}^n \sum_{j=1}^{k-1}\frac{1(x_i\leq \theta_j, y_i=j+1)}{  p_{j+1} - m_nn^{-1/2}}-\frac{1(x_i\leq \theta_j, y_i=j)}{  p_j+m_nn^{-1/2}},\]
so that 
\[|R_n(\theta, {\hat p}) - R_n(\theta)| \lesssim m_nn^{-1/2}.\]
Then, on $W^c$ the difference $R_n(\theta,\hat p) - R_n(\theta^\star, \hat p)$ can be bounded by
\begin{align*}
    R_n(\theta, {\hat p}) - R_n({\theta^\star}, {\hat p})& = R_n(\theta, {\hat p}) - R_n({\theta^\star}, {\hat p}) \\
    &+ R_n(\theta) - R_n({\theta^\star}) - [R_n(\theta) - R_n({\theta^\star})]\\
    & = [R_n(\theta, {\hat p})-R_n(\theta)] + [R_n({\theta^\star})-R_n({\theta^\star}, {\hat p})]\\
    &+R_n(\theta) - R_n({\theta^\star})\\
    & \gtrsim -m_nn^{-1/2} + R_n(\theta) - R_n({\theta^\star}).
\end{align*}
However, Lemma~\ref{lem:num} lower bounds $R_n(\theta) - R_n({\theta^\star})$ by $K(M_n\eps_n)^\gamma$ for some $K>0$ for all $\|\theta-{\theta^\star}\|>M_n\eps_n$.  Therefore, since $m_n$ is arbitrary and $\eps_n^\gamma \gtrsim n^{-1/2}$ we obtain the same bound as in Lemma~\ref{lem:num}, that is, for some $K'>0$
\[ P\Bigl(\sup_{\|\theta-{\theta^\star}\| > M_n\eps_n} \{R_n({\theta^\star}, {\hat p}) - R_n(\theta, {\hat p}) \} > -K' (M_n\eps_n)^{\gamma} \Bigr) \to 0, \quad \text{as $n \to \infty$}. \]
The rest of the proof proceeds exactly as in the proof of Theorem~4.1 a).
\end{proof}

\section{Proof of Theorem 4.2}
\subsection{Preliminary results}
Proposition~\ref{prop:gibbs_cons} below is reproduced from \citet{syring.2017}.
\begin{assump}
\label{assump:cons}
\begin{align}
\label{conds:Gibbs_cons}
&\sup_{\theta \in\Theta} \left|R_n(\theta)-R(\theta) \right|\rightarrow 0 \textrm{ in }P^n-\textrm{probability};\\
&\sup_{\|\theta-\theta^\star\|> \epsilon}R(\theta^\star)-R(\theta)<-\delta(\eps);\\
& \Pi(\{\theta: \omega R(\theta) - \omega R(\theta^\star)|\leq \alpha \})\gtrsim e^{-n\alpha}.
\end{align}    
\end{assump}
\begin{proposition}
\label{prop:gibbs_cons}
If Assumption~\ref{assump:cons} holds, then $\Pi_n(\{\theta: \|\theta-\theta^\star\| > \eps\}) \rightarrow 0$ in $P^n-$probability for any $\eps>0$.
\end{proposition}
\begin{proof}

Write the Gibbs posterior probability of the complement of the set $A = \{\theta\in\Theta: \|\theta-\theta^\star\|<\epsilon\}$,
\[\Pi_n(A^c) = \frac{\int_{A^c} \exp(-\omega R_n(\theta))d\Pi(\theta)}{\int_{\Theta} \exp(-\omega R_n(\theta))d\Pi(\theta)}\]
for some $\epsilon > 0$.    

First, bound the denominator from below as follows.  We assume the risk function $R(\theta^\star):=E(\ell_{\theta^\star})\geq 0$, but if it is bounded below by a negative number we may implicitly add an arbitrary constant to the loss function so that $R(\theta^\star)$ is positive.  Multiply the denominator, denoted $D_n$, by $e^{n(\omega R(\theta^\star)+\alpha)}$ for a positive constant $\alpha$,
\[e^{n(\omega R(\theta^\star)+\alpha)}D_n = e^{n(\omega R(\theta^\star)+\alpha)}\int_{\Theta}\exp(-n\omega R_{n}(\theta))d\Pi(\theta).\]
Bound this product from below by restricting the domain of integration,
\[e^{n(\omega R(\theta^\star)+\alpha)}D_n \geq \int_{\{\theta: \omega R(\theta)-\omega R(\theta^\star)\leq \alpha/2\}}\exp[-n\omega (R_{n}(\theta)-\omega R(\theta^\star)-\alpha)]d\Pi(\theta)\]
Add and subtract $R(\theta)$ in the exponent of the integrand, and apply the inequality in the domain of integration to get
\begin{align*}
&\int_{\{\theta: \omega R(\theta) - \omega R(\theta^\star)\leq \alpha/2 \}}\exp(-n[\omega R_n(\theta)- \omega R(\theta) + \omega R(\theta) - \omega R(\theta^\star) - \alpha])d\Pi(\theta)\\
&\geq e^{n\alpha/2}\int_{\{\theta: \omega R(\theta) - \omega R(\theta^\star)\leq \alpha/2 \}}\exp(-n[\omega R_n(\theta)- \omega R(\theta)])d\Pi(\theta).
\end{align*}
Since the above integrand is non-negative, use Fatou's Lemma to evaluate the limit
\begin{align*}
&\lim \inf_{n\rightarrow \infty} \int_{\{\theta: \omega R(\theta) - \omega R(\theta^\star)\leq \alpha/2 \}}\exp(-n[\omega R_n(\theta) - \omega R(\theta)])d\Pi(\theta) \\
&\geq \int_{\{\theta: \omega R(\theta) - \omega R(\theta^\star)\leq \alpha/2 \}} \lim \inf_{n\rightarrow \infty} \exp(-n[\omega R_n(\theta) - \omega R(\theta)])d\Pi(\theta)\\
&\geq \int_{\{\theta: \omega R(\theta) - \omega R(\theta^\star)\leq \alpha/2 \}} \exp(-\lim \sup_{n\rightarrow \infty} n\omega |R_n(\theta) - R(\theta)|)d\Pi(\theta)
\end{align*}
With this limit the denominator may be bounded from below by
\[D_n\geq e^{-n\omega \delta}\Pi(\{\theta: \omega R(\theta) - \omega R(\theta^\star)\leq \alpha/2 \})\gtrsim e^{-n(\omega\delta+\alpha/2)} \]
in $P^n-$probability and where $\delta>0$ vanishes as $n\rightarrow\infty$.  Since $\alpha>0$ is arbitrary $e^{n(\omega R(\theta^\star)+\alpha)}D_n$ diverges in $P^n-$probability as $n\rightarrow \infty$.  Hence, $D_n$ is bounded below by $Ce^{-n(\omega R(\theta^\star)+\alpha)}$ in $P^n-$probability for some $C>0$.

Next, bound the numerator from above.  Write the numerator of $\Pi_n(A)$ as
\[N_n(A) = \int_{\{\theta: \|\theta-\theta^\star\| > \epsilon\}} \exp(-n[\omega R_n(\theta)])d\Pi(\theta).\]
Add and subtract $R(\theta)$ from the exponent in the numerator to obtain
\[\int_{\{\theta: \|\theta-\theta^\star\| > \epsilon\}} \exp(-n\omega[R_n(\theta) - R(\theta) + R(\theta)])d\Pi(\theta).\]
By \eqref{conds:Gibbs_cons}, $R_n(\theta) - R(\theta)$ can be bounded uniformly over the set of integration by $\delta>0$ in $P-$probability and $R(\theta)>R(\theta^\star) + \eta$ for some $\eta(\epsilon) > 0$.  Then,
\[N_n(A) \leq e^{-n\omega[-\delta + R(\theta^\star) + \eta]}.\]

Combining the bounds on numerator and denominator, 
\begin{align*}
\Pi_n(A) = \frac{N_n(A)}{D_n} &\lesssim \frac{e^{-n\omega[-\delta + R(\theta^\star) + \eta]}}{e^{-n(\omega R(\theta^\star)+\alpha)}}\\
&= e^{n\omega\delta + n\alpha - n\omega\eta}.
\end{align*}
As $n\rightarrow \infty$, $\delta$ vanishes, but $\eta>0$ is a fixed value dependent on $\epsilon$.  So, if $\alpha$ is chosen so that $\alpha < \omega\eta/2$ the upper bound is no larger than $e^{-n\omega\eta/2}$ for all sufficiently large $n$, and vanishes in $P^n-$probability as $n \rightarrow \infty$.   
\end{proof}

\subsection{Proof of Theorem~4.2}
\label{proof:thm2}

We prove Theorem~4.2 by checking the three conditions in Assumption~\ref{assump:cons} and applying Proposition~\ref{prop:gibbs_cons}.  Assume the probabilities $p_1$ and $p_{-1}$ are known; we can remove this assumption by essentially the same argument as in the proof of Theorem~4.1 b). 

First, by the triangle inequality, (2), and our definition of $d$ we have
\begin{align*}
\|{\beta}^\top {B_d} - \theta^\star\| &\leq \|{\beta}^\top {B_d} - {{\beta_d^\star}}^\top {B_d}\|+\|{{\beta_d^\star}}^\top {B_d} - \theta^\star\|\\
&\leq \|{\beta}^\top {B_d} - {{\beta_d^\star}}^\top {B_d}\|+\eps/2.
\end{align*}
Therefore the set $A_n(\eps)$ is contained in the set $\{\beta\in\mathbb{R}^d: \|{\beta}^\top {B_d} - {{\beta_d^\star}}^\top {B_d}\| > \eps/2\}$, and it is sufficient to show $\Pi_n\{\beta\in\mathbb{R}^d: \|{\beta}^\top {B_d} - {{\beta_d^\star}}^\top {B_d}\| > \eps\}]\rightarrow 0$ in $P^n-$probability for any $\eps>0$.

We begin by verifying (12) using the lower bound on a prior probability given in Assumption 4.2 i.  Define the function
\[\ell_{\beta}(x, y, z) = \frac{1(x\leq {\beta}^\top {B_d(z)}, \, y=1)}{p_1} - \frac{1(x\leq {\beta}^\top {B_d(z)}, \, y=-1)}{p_{-1}},\]
and its expectation $R({\beta}) = E(\ell_{\beta}(X, Y, Z))$. Then, we want to show
\[\Pi(\{\theta:\omega R(\beta) - R({\beta_d^\star})\leq \alpha\})\gtrsim e^{-n\alpha}.\]
Assumption~\ref{assump:thm2} ii. implies
\[R(\beta) - R({\beta_d^\star})\lesssim \|{\beta}^\top {B_d}-{{\beta_d^\star}}^\top{B_d}\|,\]
where
\[\|{\beta}^\top {B_d}-{{\beta_d^\star}}^\top{B_d}\|:=\int_\mathbb{Z} \left|{\beta}^\top {B_d(z)}-{{\beta_d^\star}}^\top{B_d(z)}\right|dz.\]
Define the $\sup-$norm balls
\[C_d(\eps):=\{\beta\in\mathbb{R}^d: \|{\beta}^\top {B_d} - {{\beta_d^\star}}^\top {B_d}\|_\infty<\eps\}.\]
Then, $\beta\in C_d(\eps)$ implies $\|{\beta}^\top {B_d}-{{\beta_d^\star}}^\top {B_d}\|< \eps$.  And, using Assumption~\ref{assump:thm2} i. and the fact that
\[\|{\beta}^\top {B_d} - {{\beta_d^\star}}^\top {B_d}\|_\infty\lesssim d\|\beta-{\beta_d^\star}\|_2,\]
we have the following prior bound:
\begin{align*}
    \Pi(C_d(\eps)) &\geq \Pi(\{\beta:\|\beta-{\beta_d^\star}\|_2 \leq c_1d^{-1}\eps\})\\
    &\geq e^{-c_2d\log(d/(c_1\eps))}
\end{align*}
for some constants $c_1, \,c_2>0$; see also the proof of Theorem~1 in \citet{shen.ghosal}.  Working backwards, we have shown
\[\Pi(\{\beta\in\RR^d:R({\beta}) - R({\beta_d^\star}) < \eps\})\gtrsim e^{-c_2d\log(d/(c_1\eps))}>e^{-n\eta}\]
for any $\eta>0$ for all large enough $n$.  

Next, we note that (11) is implied by Assumption~\ref{assump:thm2} ii.  By definition of the norm $\|\cdot\|$ and by lower bounding the densities $f$ and $g$ away from zero we have
\begin{align*}
R({\beta}) - R({\beta_d^\star}) &=\int_\mathcal{X}\int_0^1 |{\beta}^\top {B_d}(z)-{{\beta_d^\star}}^\top {B_d}(z)|f(x)dxg(z)dz\\
&\gtrsim \int_0^1 |{\beta}^\top {B_d}(z)-{{\beta_d^\star}}^\top {B_d}(z)|dz\\
&=: \|{\beta}^\top {B_d}-{{\beta_d^\star}}^\top {B_d}\|.\end{align*}
Therefore,
\[\|{\beta}^\top {B_d}-{{\beta_d^\star}}^\top {B_d}\| > \eps \Rightarrow R({\beta}) - R({\beta_d^\star})\gtrsim \eps.\]

Last, we verify (10).  Define the class of functions $\mathcal{F}:=\{\sign({\beta}^\top {B_d}(z) - x): \beta\in \mathbb{R}^d, x\in \mathcal{X}\}$.  A classic result about linear classifiers says $\mathcal{F}$ has Vapnik-Chervonenkis (VC) dimension less than $d+2$; see, e.g., Example 19.17 in \citet{vaart}.  Let $f_{\beta}(X,Z):=\sign({\beta}^\top {B_d}(Z) - X)$ denote an element of $\mathcal{F}$. The VC Inequality \citep[see, e.g., Theorem 12.5 in ][]{devroye} along with the Sauer-Shelah Lemma \citep{sauer, shelah} provides the uniform probability bound
\[P\left(\sup_{\beta}\left|\tfrac1n\sum_{i=1}^n 1[Y_i\ne f_{\beta}(X_i,Z_i)] - E\{1[Y\ne f_{\beta}(X,Z)]\}\right|>\eps\right)\lesssim n^{d+1}e^{-\frac{n\eps^2}{32}}.\]
With very minor modifications the above bound applies to the functions $1[Y_i = 1, f_{\beta}(X_i,Z_i) = 1]$ and $1[Y_i = -1, f_{\beta}(X_i,Z_i) = 1]$.  Therefore, define the functions
\begin{align*}
    R_n(\beta) &= \frac{1}{p_1}R_{n,1}(\beta) - \frac{1}{p_{-1}}R_{n,2}(\beta)\\
    &= \frac{1}{p_1}\frac{1}{n}\sum_{i=1}^n1[Y_i = 1, f_{\beta}(X_i,Z_i) = 1] - \frac{1}{p_{-1}}\frac{1}{n}\sum_{i=1}^n1[Y_i = -1, f_{\beta}(X_i,Z_i) = 1],
\end{align*}
where $E[R_n(\beta)] = R(\beta) =: \frac{1}{p_1}R_1(\beta)-\frac{1}{p_{-1}}R_2(\beta)$, correspondingly. Then, 
\begin{align*}P(\sup_{\beta}|R_n(\beta) - R(\beta)|>\eps) &\leq P(\sup_{\beta}|R_{n,1}(\beta) - R_1(\beta)|>p_1\eps/2)\\&+P(\sup_{\beta}|R_{n,2}(\beta) - R_2(\beta)|>p_{-1}\eps/2).\end{align*}
The two probabilities on the right hand side of the above display are bounded by the VC Inequality and the Sauer-Shelah Lemma, so that the left hand side has the following bound:
\[P(\sup_{\beta}|R_n(\beta) - R(\beta)|>\eps)\lesssim n^{d+1}e^{-n\eps^2\max\{p_1^2, \, p_{-1}^2\}/32}.\]
The bound vanishes as $n\rightarrow\infty$, verifying (10).

Finally, if the probabilities $p_1$ and/or $p_{-1}$ are unknown, then the corresponding objective function 
\[R_n({\beta},{\hat p}):=\frac{1}{\hat p_1}\frac{1}{n}\sum_{i=1}^n1[Y_i = 1, f_{\beta}(X_i,Z_i) = 1] - \frac{1}{\hat p_{-1}}\frac{1}{n}\sum_{i=1}^n1[Y_i = -1, f_{\beta}(X_i,Z_i) = 1]\] is within $m_nn^{-1/2}$ of $R_n({\beta})$ for any diverging sequence $m_n>0$ with $P^n-$probability tending to $1$, by the same argument as in the proof of Theorem~4.1 b).  Applying the VC Inequality and Sauer-Shelah Lemma, we have
\[P(\sup_{\beta}|R_n(\beta;\hat p) - R(\beta)|>\eps-m_nn^{-1/2})\lesssim n^{d+1}e^{-n\eps^2\max\{p_1^2, \, p_{-1}^2\}/32}+o(1),\]
and, for all large enough $n$ such that $\eps-m_nn^{-1/2}>\eps/2$
\[P(\sup_{\beta}|R_n(\beta;\hat p) - R(\beta)|>\eps/2)\rightarrow 0,\]
as $n\rightarrow\infty$.  Therefore, the result of Theorem~\ref{thm:2} holds whether $p_1$ and $p_{-1}$ are known or replaced by the corresponding sample proportions.

\end{document}